\documentclass[11pt]{amsart}

\usepackage{amsmath,amssymb,amsthm,mathtools}
\usepackage{microtype}
\usepackage[shortlabels]{enumitem}
\usepackage{booktabs}
\usepackage{longtable}
\usepackage{listings}
\usepackage{xurl}
\usepackage[hidelinks,hypertexnames=false]{hyperref}

\allowdisplaybreaks[2]
\newtheorem{theorem}{Theorem}[section]
\newtheorem{lemma}[theorem]{Lemma}
\newtheorem{proposition}[theorem]{Proposition}
\newtheorem{corollary}[theorem]{Corollary}
\theoremstyle{definition}
\newtheorem{definition}[theorem]{Definition}
\theoremstyle{remark}
\newtheorem{remark}[theorem]{Remark}

\newcommand{\cP}{\mathcal{P}}
\newcommand{\cE}{\mathcal{E}}

\newcommand{\V}{V}
\newcommand{\E}{E}
\newcommand{\del}{\delta}

\newcommand{\doi}[1]{\href{https://doi.org/#1}{doi:\nolinkurl{#1}}}

\title[3-critical 3-graphs of minimum degree 7]{On an Erd\H{o}s--Lov\'asz problem: 3-critical 3-graphs of minimum degree~7}

\author{Ruiliang Li}
\address{Tsinghua University, Beijing 100084, China}
\email{lirl23@mails.tsinghua.edu.cn}


\date{December 19, 2025}

\begin{document}

\begin{abstract}
Erd\H{o}s and Lov\'asz asked whether there exists a ``$3$-critical'' $3$-uniform hypergraph in which every vertex has degree at least $7$.
The original formulation does not specify what $3$-critical means, and two non-equivalent notions have appeared in the literature and in later discussions of the problem.
In this paper we resolve the question under both interpretations.
For the transversal interpretation (criticality with respect to the transversal number), we prove that a $3$-uniform hypergraph $H$ with $\tau(H)=3$ and $\tau(H-e)=2$ for every edge $e$ has at most $10$ edges; in particular, $\del(H)\le 6$, and this bound is sharp, witnessed by the complete $3$-graph $K^{(3)}_5$.
For the chromatic interpretation (criticality with respect to weak vertex-colourings), we give an explicit $3$-uniform hypergraph on $9$ vertices with $\chi(H)=3$ and minimum degree $\del(H)=7$ such that deleting any single edge or any single vertex makes it $2$-colourable.
The criticality of the example is certified by explicit witness $2$-colourings listed in the appendices, together with a short verification script.
\end{abstract}

\maketitle

\section{Introduction}\label{sec:intro}

A (weak) vertex-colouring of a hypergraph $H=(\V,\E)$ is an assignment of colours to the vertices such that no edge is monochromatic.
This is the standard extension of graph colouring and is closely tied to classical questions in extremal set theory (e.g.\ Property~B) and probabilistic combinatorics.
In their influential paper on $3$-chromatic hypergraphs, Erd\H{o}s and Lov\'asz~\cite{ErdosLovasz1975} investigated structural and extremal properties of non-$2$-colourable uniform hypergraphs and posed a number of problems that have continued to resonate in the subject.
Among them is the following local question:

\medskip
\noindent\emph{Does there exist a $3$-critical $3$-uniform hypergraph in which every vertex has degree at least $7$?}
\medskip

In the graph case, $k$-criticality is unambiguous and forces strong degree constraints; in particular, every $k$-critical graph has minimum degree at least $k-1$.
For hypergraphs, several competing notions of ``critical'' coexist in the literature, and the correct interpretation of Erd\H{o}s--Lov\'asz's question is therefore not immediate.
One usage (common in extremal set theory) is formulated in terms of hitting sets: a hypergraph is ``critical'' if its transversal number is exactly $t$ but deleting any edge reduces the transversal number.
Another usage (closer to graph colouring) is colour-criticality: a hypergraph is $k$-critical if its chromatic number is $k$ but every proper subhypergraph is $(k-1)$-colourable.
Both notions appear in standard references, and the ambiguity is explicitly noted in contemporary problem compilations; see, for instance,~\cite{ErdosProblems834}.
Since Erd\H{o}s and Lov\'asz posed the question in a paper centred around chromatic phenomena, the chromatic interpretation is arguably the most natural; nevertheless, the transversal interpretation is sufficiently widespread that it merits a definitive treatment as well.

The purpose of this paper is to resolve Erd\H{o}s--Lov\'asz's question under \emph{both} interpretations.
We use the terminology introduced in Section~\ref{sec:terminology} and state our main results here.

\begin{theorem}[Transversal interpretation]\label{thm:intro-transversal}
Let $H$ be a $3$-uniform hypergraph that is \emph{$\tau$-critical of order $3$}, i.e.\ $\tau(H)=3$ and $\tau(H-e)=2$ for every $e\in \E(H)$.
Then $|\E(H)|\le 10$ and consequently $\del(H)\le 6$.
Moreover, equality $\del(H)=6$ is attained by the complete $3$-graph $K^{(3)}_5$.
\end{theorem}

\begin{theorem}[Chromatic interpretation]\label{thm:intro-chromatic}
There exists a $3$-uniform hypergraph $H$ that is \emph{critically $3$-chromatic} (under weak colourings) and satisfies $\del(H)\ge 7$.
In fact, there is such an $H$ on $9$ vertices with $\del(H)=7$.
\end{theorem}

Theorem~\ref{thm:intro-transversal} shows that the transversal interpretation yields a sharp negative answer to the minimum-degree requirement.
In contrast, Theorem~\ref{thm:intro-chromatic} gives a sharp positive answer under the chromatic-number interpretation (as stated on the Erd\H{o}s Problems website~\cite{ErdosProblems834}).
Our example is simultaneously edge-minimal and vertex-minimal for non-$2$-colourability (Property~B), and we provide explicit certificates in Appendix~\ref{app:certificates}; a short brute-force verification script is included in Appendix~\ref{app:verification}.

\section{Terminology and basic notions}\label{sec:terminology}

\subsection{Hypergraphs, uniformity, and degrees}

A \emph{hypergraph} $H$ is a pair $H=(\V(H),\E(H))$, where $\V(H)$ is a finite set of vertices and $\E(H)$ is a family of subsets of $\V(H)$ called \emph{edges}.
We always assume that $H$ is \emph{simple} (no repeated edges) and has no empty edge.
The hypergraph $H$ is \emph{$r$-uniform} if $|e|=r$ for all $e\in \E(H)$.
In particular, a $3$-uniform hypergraph is often called a \emph{$3$-graph}.

For a vertex $v\in \V(H)$, the \emph{degree} $d_H(v)$ is the number of edges containing $v$:
\[
d_H(v):=\bigl|\{e\in \E(H): v\in e\}\bigr|.
\]
The \emph{minimum degree} is $\del(H):=\min_{v\in \V(H)} d_H(v)$.

If $e\in \E(H)$, we write
\[
H-e := \bigl(\V(H),\,\E(H)\setminus\{e\}\bigr).
\]
If $v\in \V(H)$, then $H-v$ denotes the induced subhypergraph obtained by deleting $v$ and all incident edges:
\[
H-v := H\bigl[\V(H)\setminus\{v\}\bigr].
\]
When discussing vertex-criticality we may disregard isolated vertices, since they play no role for weak colourability but obstruct vertex-minimality; this is in line with standard conventions in criticality theory.\cite{KostochkaWoodall2000}

\subsection{Weak colourings, Property~B, and chromatic criticality}

A \emph{(vertex) $k$-colouring} of a hypergraph $H$ is a map $\varphi:\V(H)\to[k]:=\{1,2,\dots,k\}$.
We call such a colouring \emph{proper} (or \emph{weakly proper}) if no edge is monochromatic, i.e.
\[
\forall e\in \E(H)\qquad \bigl|\varphi(e)\bigr|\ge 2.
\]
This is the standard notion of hypergraph colouring in the set-system and Property~B literature; see, for instance,~\cite{Berge1989,BujtasTuzaVoloshin2015,DurajGutowskiKozik2018}.
It differs from \emph{strong} (or \emph{rainbow}) colourings, where all vertices in every edge must receive distinct colours; strong colourings have their own criticality theory and degree questions (see, e.g.,~\cite{KostochkaWoodall2000}).

The \emph{chromatic number} $\chi(H)$ is the smallest $k$ such that $H$ admits a proper $k$-colouring.
A hypergraph is \emph{$2$-colourable} if $\chi(H)\le 2$.
In the classical terminology of Erd\H{o}s--Hajnal and their successors, $2$-colourability is equivalent to \emph{Property~B}.\cite{DurajGutowskiKozik2018}

Several variants of criticality exist for hypergraphs (edge-minimal, vertex-minimal, or both); see Toft's foundational paper and subsequent work, as well as surveys.\cite{Toft1974,AbbottHare1999,KostochkaSurvey2006}
For the purposes of this paper we adopt the following version, which matches the chromatic interpretation explicitly recorded in~\cite{ErdosProblems834}.

\begin{definition}[Critically $3$-chromatic (weak)]\label{def:critical-3chrom}
A hypergraph $H$ is \emph{critically $3$-chromatic} if
\begin{equation}\label{eq:def-critical-3chrom}
\begin{aligned}
\chi(H)&=3,\\
\chi(H-e)&\le 2 \quad \text{for every } e\in \E(H),\\
\chi(H-v)&\le 2 \quad \text{for every } v\in \V(H).
\end{aligned}
\end{equation}
Equivalently, $H$ is not $2$-colourable, but becomes $2$-colourable after deleting any single edge or any single vertex.
\end{definition}

The condition \eqref{eq:def-critical-3chrom} is strictly stronger than the frequently used \emph{edge-critical} condition (requiring only $\chi(H-e)\le 2$ for all $e$), and also stronger than the \emph{vertex-critical} condition (requiring only $\chi(H-v)\le 2$ for all $v$).
Such distinctions are standard already for graphs and appear in the hypergraph setting as well.\cite{KostochkaWoodall2000,Toft1974}

A convenient feature of edge-criticality in the case $k=3$ is that it admits short checkable certificates in the form of explicit $2$-colourings.

\begin{lemma}[Edge-deletion certificates]\label{lem:edge-cert}
Let $H$ be a hypergraph with $\chi(H)\ge 3$, and let $e\in \E(H)$.
Then $\chi(H-e)\le 2$ if and only if there exists a $2$-colouring $\varphi:\V(H)\to\{1,2\}$ such that $e$ is the \emph{unique} monochromatic edge under $\varphi$.
\end{lemma}

\begin{proof}
If $\chi(H-e)\le 2$, fix a proper $2$-colouring $\varphi$ of $H-e$.
Under the same colouring of $\V(H)$, every edge of $H$ except possibly $e$ is non-monochromatic.
If $e$ were also non-monochromatic, then $\varphi$ would be a proper $2$-colouring of $H$, contradicting $\chi(H)\ge 3$.
Hence $e$ must be monochromatic, and it is the only monochromatic edge.

Conversely, if there exists a $2$-colouring $\varphi$ for which $e$ is the unique monochromatic edge, then every edge of $H-e$ is non-monochromatic under $\varphi$, so $\chi(H-e)\le 2$.
\end{proof}

In particular, for a critically $3$-chromatic hypergraph one may certify edge-minimality by exhibiting, for every $e\in \E(H)$, a $2$-colouring with $e$ as the unique monochromatic edge.
Similarly, vertex-minimality can be certified by providing, for every $v\in \V(H)$, a proper $2$-colouring of $H-v$.

\subsection{Transversals and \texorpdfstring{$\tau$}{tau}-criticality of order 3}

A \emph{transversal} (or \emph{hitting set}) of a hypergraph $H$ is a subset $T\subseteq \V(H)$ such that $T\cap e\neq\emptyset$ for all $e\in \E(H)$.
The minimum size of a transversal is the \emph{transversal number} (or \emph{covering number}) $\tau(H)$.\cite{Berge1989}

\begin{definition}[$\tau$-critical of order $3$]\label{def:tau-critical}
A $3$-uniform hypergraph $H$ is \emph{$\tau$-critical of order $3$} if
\[
\tau(H)=3
\quad\text{and}\quad
\tau(H-e)\le 2 \ \text{ for every } e\in \E(H).
\]
\end{definition}

This matches the ``pair hits all edges except $e$'' formulation explicitly quoted in~\cite{ErdosProblems834}; we will refer to it as \emph{$3$-critical in the transversal sense}.
The parameters $\chi(H)$ and $\tau(H)$ capture different aspects of a hypergraph and lead to genuinely different criticality theories.
For example, a disjoint union of three $3$-edges has $\tau=3$ but is $2$-colourable (hence $\chi=2$), so $\tau(H)=3$ does not force $3$-chromaticity.
Conversely, critically $3$-chromatic $3$-graphs may have transversal number much larger than $3$.
Accordingly, in what follows we treat the chromatic and transversal interpretations separately and avoid using the term ``$3$-critical'' without specifying the intended meaning.

\section{The transversal interpretation: a sharp bound and the impossibility of \texorpdfstring{$\del\ge 7$}{delta>=7}}\label{sec:tau}

We first address the transversal interpretation from Section~\ref{sec:terminology}.
In that setting the Erd\H{o}s--Lov\'asz question has a clean and essentially finite answer: a $3$-graph that is $\tau$-critical of order $3$ has at most $10$ edges, and hence cannot have minimum degree at least $7$.
The main tool is the classical set-pairs inequality of Bollob\'as,\cite{Bollobas1965} which has become a standard ingredient in extremal set theory; see also Katona's short exposition.\cite{Katona1974}
We include a self-contained proof for completeness.

\subsection{Bollob\'as's set-pairs inequality}

Let $X$ be a finite set.
A \emph{set-pair system} is a family of pairs of subsets
\[
\cP=\{(A_i,B_i)\}_{i=1}^m,\qquad A_i,B_i\subseteq X.
\]

\begin{lemma}[Bollob\'as set-pairs inequality]\label{lem:bollobas}
Let $\{(A_i,B_i)\}_{i=1}^m$ be pairs of finite sets such that
\begin{equation}\label{eq:bollobas-hyp}
A_i\cap B_i=\emptyset\quad\text{for all }i,\qquad\text{and}\qquad
A_i\cap B_j\neq\emptyset\quad\text{for all }i\neq j.
\end{equation}
Then
\begin{equation}\label{eq:bollobas}
\sum_{i=1}^m \binom{|A_i|+|B_i|}{|A_i|}^{-1}\le 1.
\end{equation}
\end{lemma}

\begin{proof}
Let $Y$ be the union of all elements that appear in at least one of the sets $A_i$ or $B_i$.
Fix a uniformly random permutation $\pi$ of $Y$.
For each $i$, define the event
\[
\cE_i=\{\text{in }\pi\text{, every element of }A_i\text{ appears before every element of }B_i\}.
\]
We claim that the events $\cE_1,\dots,\cE_m$ are pairwise disjoint.
Indeed, suppose $\cE_i$ and $\cE_j$ both occur with $i\neq j$.
By \eqref{eq:bollobas-hyp} there exists $x\in A_i\cap B_j$.
Since $\cE_i$ holds, the element $x\in A_i$ must appear before all elements of $B_i$.
Since $\cE_j$ holds and $x\in B_j$, every element of $A_j$ must appear before $x$.
Again by \eqref{eq:bollobas-hyp}, choose $y\in A_j\cap B_i$.
Then $y\in A_j$ implies $y$ appears before $x$, while $y\in B_i$ implies $x$ appears before $y$ (because under $\cE_i$ all elements of $A_i$ precede all elements of $B_i$), a contradiction.
Hence $\cE_i\cap \cE_j=\emptyset$ for $i\neq j$.

Therefore
\[
\sum_{i=1}^m \mathbb{P}(\cE_i)\le 1.
\]
It remains to compute $\mathbb{P}(\cE_i)$.
In the restriction of $\pi$ to $A_i\cup B_i$, all relative orders are equally likely.
Equivalently, the positions occupied by the elements of $A_i$ among the $|A_i|+|B_i|$ positions are a uniformly random $|A_i|$-subset.
The event $\cE_i$ occurs precisely when $A_i$ occupies the first $|A_i|$ positions among $A_i\cup B_i$, which has probability
\[
\mathbb{P}(\cE_i)=\binom{|A_i|+|B_i|}{|A_i|}^{-1}.
\]
Summing over $i$ gives \eqref{eq:bollobas}.
\end{proof}

\subsection{An edge bound for \texorpdfstring{$\tau$}{tau}-critical \texorpdfstring{$3$}{3}-graphs}

\begin{theorem}\label{thm:tau-main}
Let $H$ be a $3$-uniform hypergraph that is $\tau$-critical of order $3$.
Then $|\E(H)|\le 10$.
\end{theorem}

\begin{proof}
Fix an edge $e\in \E(H)$.
By $\tau(H-e)\le 2$, there exists a set $B_e\subseteq \V(H)$ with $|B_e|\le 2$ that meets every edge of $H-e$.
We first note that in fact $|B_e|=2$.
Indeed, if $|B_e|=1$, say $B_e=\{x\}$, then $x$ meets every edge of $H-e$.
If $x\in e$, then $\{x\}$ meets every edge of $H$ and $\tau(H)=1$, contradicting $\tau(H)=3$.
If $x\notin e$, then for any $y\in e$ the $2$-set $\{x,y\}$ meets every edge of $H$ (the vertex $x$ meets all edges except possibly $e$, and $y$ meets $e$), implying $\tau(H)\le 2$, again a contradiction.
Thus $|B_e|=2$.

Next we claim that
\begin{equation}\label{eq:Be-disjoint}
B_e\cap e=\emptyset.
\end{equation}
If not, then $B_e$ meets every edge of $H-e$ by definition, and it also meets $e$.
Hence $B_e$ would be a transversal of $H$ of size $2$, contradicting $\tau(H)=3$.
So \eqref{eq:Be-disjoint} holds.

Now define, for each edge $e\in \E(H)$,
\[
A_e:=e,\qquad B_e\ \text{as above}.
\]
Then $A_e\cap B_e=\emptyset$ by \eqref{eq:Be-disjoint}.
Moreover, if $f\in \E(H)$ is a different edge ($f\neq e$), then $f$ is an edge of $H-e$, so it must meet $B_e$, i.e.\ $A_f\cap B_e=f\cap B_e\neq\emptyset$.
Thus the family of pairs $\{(A_e,B_e)\}_{e\in \E(H)}$ satisfies \eqref{eq:bollobas-hyp}.

Applying Lemma~\ref{lem:bollobas} and using $|A_e|=3$ and $|B_e|=2$ for all $e$, we obtain
\[
|\E(H)|\cdot \binom{3+2}{3}^{-1}\le 1,
\]
hence $|\E(H)|\le \binom{5}{3}=10$.
\end{proof}

\begin{remark}
The same argument yields the general pattern: for an $r$-uniform hypergraph with $\tau(H)=t$ and $\tau(H-e)\le t-1$ for all $e$, one obtains $|\E(H)|\le \binom{r+t-1}{r}$.
In our setting $(r,t)=(3,3)$, this gives $10$.
Related applications of the set-pairs method in this context appear in Tuza's work.\cite{Tuza1985}
\end{remark}

\subsection{A sharp minimum-degree consequence}

\begin{corollary}\label{cor:tau-mindeg}
If $H$ is a $3$-uniform hypergraph that is $\tau$-critical of order $3$, then $\del(H)\le 6$.
In particular, there is no such $H$ with $\del(H)\ge 7$.
\end{corollary}

\begin{proof}
By Theorem~\ref{thm:tau-main}, $|\E(H)|\le 10$.
Since $H$ is $3$-uniform, the sum of degrees satisfies
\[
\sum_{v\in \V(H)} d_H(v)=3|\E(H)|\le 30.
\]
On the other hand, $\tau(H)=3$ forces $|\V(H)|\ge 5$.
Indeed, if $|\V(H)|\le 4$, then every $2$-subset of $\V(H)$ meets every $3$-subset, so any $3$-uniform hypergraph on at most $4$ vertices has $\tau\le 2$, contradicting $\tau(H)=3$.
Therefore,
\[
\frac{1}{|\V(H)|}\sum_{v\in \V(H)} d_H(v)\le \frac{30}{5}=6.
\]
Since $\del(H)$ is at most the average degree, we obtain $\del(H)\le 6$.
\end{proof}

\subsection{Sharpness: the extremal example}

Let $K^{(3)}_5$ denote the complete $3$-uniform hypergraph on five vertices.
Then it has $10$ edges and is $6$-regular.

\begin{proposition}\label{prop:K5tau}
The hypergraph $K^{(3)}_5$ is $\tau$-critical of order $3$ and satisfies $\del(K^{(3)}_5)=6$.
\end{proposition}

\begin{proof}
First, $\tau(K^{(3)}_5)=3$ because no $2$-subset of $[5]$ is a transversal: given distinct $x,y\in[5]$, the complementary $3$-set $[5]\setminus\{x,y\}$ is an edge disjoint from $\{x,y\}$.
On the other hand, any $3$-subset of $[5]$ meets every $3$-subset, so $\tau\le 3$ and hence $\tau=3$.

Now fix an edge $e\in \binom{[5]}{3}$, and let $\{u,v\}=[5]\setminus e$ be its complementary pair.
Then $\{u,v\}$ meets every $3$-subset of $[5]$ other than $e$ (because any other $3$-subset must contain at least one of $u$ or $v$), so $\tau(K^{(3)}_5-e)\le 2$.
As in the proof of Theorem~\ref{thm:tau-main}, $\tau(K^{(3)}_5-e)$ cannot be $1$, so $\tau(K^{(3)}_5-e)=2$.
Thus $K^{(3)}_5$ is $\tau$-critical of order $3$.
Finally, $\del(K^{(3)}_5)=\binom{4}{2}=6$ is immediate.
\end{proof}

Corollary~\ref{cor:tau-mindeg} shows that, under the transversal interpretation of ``$3$-critical'', the minimum-degree requirement $\del\ge 7$ is impossible.
Consequently, any affirmative resolution of the Erd\H{o}s--Lov\'asz question must rely on the chromatic-number interpretation, which we address next.

\section{The chromatic interpretation: a critically \texorpdfstring{$3$}{3}-chromatic \texorpdfstring{$3$}{3}-graph with \texorpdfstring{$\del=7$}{delta=7}}\label{sec:chromatic}

We now turn to the chromatic-number interpretation from Section~\ref{sec:terminology}.
In the language of Property~B, we seek a $3$-uniform hypergraph that is minimally non-$2$-colourable under both edge and vertex deletion, while simultaneously forcing all degrees to be at least $7$.
Colour-critical hypergraphs have been studied since the work of Toft and others,\cite{Toft1974,AbbottHare1999,KostochkaSurvey2006}
but the particular combination ``$3$-uniform + weak colouring + minimum degree constraint + edge- and vertex-criticality'' does not seem to be covered by general structural results.
We therefore give an explicit compact example.

\subsection{The construction}\label{subsec:construction}

\begin{theorem}\label{thm:construction}
There exists a critically $3$-chromatic $3$-uniform hypergraph $H$ with $\del(H)\ge 7$.
In fact, the $3$-graph $H$ defined below on $9$ vertices satisfies $\del(H)=7$.
\end{theorem}

\begin{proof}
Let $H$ be the $3$-uniform hypergraph with vertex set
\[
\V(H)=[9]:=\{1,2,3,4,5,6,7,8,9\}
\]
and edge set
\begin{equation}\label{eq:H-edges}
\begin{aligned}
\E(H)=\{&
\{1,2,3\},\{1,2,9\},\{1,3,8\},\{1,4,6\},\{1,4,8\},\{1,4,9\},\\
&
\{1,5,7\},\{1,5,8\},\{1,5,9\},\{1,6,7\},\{2,3,6\},\{2,3,7\},\\
&
\{2,4,9\},\{2,5,9\},\{2,6,7\},\{3,4,8\},\{3,5,8\},\{3,6,7\},\\
&
\{4,6,8\},\{4,6,9\},\{5,7,8\},\{5,7,9\}\}.
\end{aligned}
\end{equation}
For later reference, the first ten edges in \eqref{eq:H-edges} are exactly the edges incident with vertex~$1$, while the remaining twelve edges lie entirely in $\{2,3,4,5,6,7,8,9\}$.
Appendix~\ref{app:edges} repeats \eqref{eq:H-edges} in a machine-readable format.

We will show that $H$ is critically $3$-chromatic and has minimum degree $7$ by proving Lemmas~\ref{lem:degrees}--\ref{lem:3-colorable} and Propositions~\ref{prop:edge-critical}--\ref{prop:vertex-critical}.
\end{proof}

\subsection{Degree computation}

\begin{lemma}\label{lem:degrees}
The hypergraph $H$ defined in \eqref{eq:H-edges} satisfies $d_H(1)=10$ and $d_H(v)=7$ for every $v\in\{2,\dots,9\}$.
In particular, $\del(H)=7$.
\end{lemma}

\begin{proof}
From \eqref{eq:H-edges}, vertex~$1$ lies in the ten edges listed on the first line, so $d_H(1)=10$.
For $v\in\{2,\dots,9\}$, a direct count from \eqref{eq:H-edges} shows that exactly seven edges contain~$v$.
For instance,
\[
d_H(2)=|\{123,129,236,237,249,259,267\}|=7,
\]
and the remaining cases are analogous.
Thus $\del(H)=7$.
\end{proof}

\subsection{\texorpdfstring{$H$}{H} is not \texorpdfstring{$2$}{2}-colourable}\label{subsec:not-2-colorable}

We next give a self-contained proof that $H$ has no proper $2$-colouring (i.e.\ fails Property~B), and hence $\chi(H)\ge 3$.
The proof is designed to be checkable directly from \eqref{eq:H-edges}; it exploits the separation between edges through vertex~$1$ and the ``core'' on $\{2,\dots,9\}$.

\begin{lemma}\label{lem:not-2-colorable}
The $3$-graph $H$ defined in \eqref{eq:H-edges} is not $2$-colourable.
\end{lemma}

\begin{proof}
Suppose, for a contradiction, that $\varphi:\V(H)\to\{0,1\}$ is a proper $2$-colouring of $H$ (no monochromatic edge).
By swapping the colours if necessary, we may assume
\begin{equation}\label{eq:phi1}
\varphi(1)=0.
\end{equation}
Let
\[
\begin{aligned}
Z&:=\{v\in\{2,\dots,9\}:\varphi(v)=0\},\\
B&:=\{v\in\{2,\dots,9\}:\varphi(v)=1\}=\{2,\dots,9\}\setminus Z.
\end{aligned}
\]

\smallskip
\noindent\emph{Step 1: $|Z|\le 3$.}
Consider the graph $G$ on vertex set $\{2,\dots,9\}$ where $\{x,y\}$ is an edge of $G$ precisely when $\{1,x,y\}$ is an edge of $H$.
From the first line of \eqref{eq:H-edges}, we have
\begin{equation}\label{eq:G-edges}
\begin{aligned}
\E(G)=\{&
\{2,3\},\{2,9\},\{3,8\},\{4,6\},\{4,8\},\\
&
\{4,9\},\{5,7\},\{5,8\},\{5,9\},\{6,7\}\}.
\end{aligned}
\end{equation}
Because of \eqref{eq:phi1}, for every edge $\{1,x,y\}\in \E(H)$ we cannot have $\varphi(x)=\varphi(y)=0$, otherwise $\{1,x,y\}$ would be monochromatic $0$.
Equivalently, $Z$ is an independent set of $G$.

We claim that $\alpha(G)=3$, and hence $|Z|\le 3$.
Indeed, if $8\in Z$, then by \eqref{eq:G-edges} we have $3,4,5\notin Z$, so $Z\subseteq\{2,6,7,8,9\}$ with $8$ fixed.
Within $\{2,6,7,9\}$, the edges $\{2,9\}$ and $\{6,7\}$ again appear in \eqref{eq:G-edges}, so we may choose at most two additional vertices; thus $|Z|\le 3$.
The same argument applies if $9\in Z$.
If neither $8$ nor $9$ lies in $Z$, then $Z\subseteq\{2,3,4,5,6,7\}$.
The induced subgraph on $\{4,5,6,7\}$ is the path $4$--$6$--$7$--$5$ by \eqref{eq:G-edges}, so it has independence number $2$, and the remaining edge $\{2,3\}$ contributes at most one more vertex; hence again $|Z|\le 3$.
This proves the claim, and therefore
\begin{equation}\label{eq:B-size}
|B|=8-|Z|\ge 5.
\end{equation}

\smallskip
\noindent\emph{Step 2: every $5$-subset of $\{2,\dots,9\}$ contains an edge of $H$ not using vertex $1$.}
Let $H_0:=H[\{2,\dots,9\}]$.
Its edge set consists of the last twelve edges in \eqref{eq:H-edges}:
\begin{equation}\label{eq:H0-edges}
\begin{aligned}
\E(H_0)=\{&
\{2,3,6\},\{2,3,7\},\{2,4,9\},\{2,5,9\},\{2,6,7\},\\
&
\{3,4,8\},\{3,5,8\},\{3,6,7\},\{4,6,8\},\{4,6,9\},\\
&
\{5,7,8\},\{5,7,9\}\}.
\end{aligned}
\end{equation}
We claim that $H_0$ has no independent set of size $5$.
Let $S\subseteq\{2,\dots,9\}$ with $|S|=5$; we show that $S$ contains an edge from \eqref{eq:H0-edges}.
Set $A:=\{2,3,6,7\}$.
Every $3$-subset of $A$ is an edge of $H_0$, namely
\[
\binom{A}{3}=\bigl\{\{2,3,6\},\{2,3,7\},\{2,6,7\},\{3,6,7\}\bigr\}\subseteq \E(H_0).
\]

We distinguish cases according to $S\cap\{8,9\}$.

\smallskip
\noindent\emph{Case 1: $S\cap\{8,9\}=\emptyset$.}
Then $S\subseteq A\cup\{4,5\}$, and since $|S|=5$ we have $|S\cap A|\ge 3$.
Hence $S$ contains a $3$-subset of $A$, and therefore contains an edge of $H_0$.

\smallskip
\noindent\emph{Case 2: $|S\cap\{8,9\}|=1$.}
By symmetry we may assume $8\in S$ and $9\notin S$.
If $|S\cap A|\ge 3$ we are done, so assume $|S\cap A|\le 2$.
Then $S\setminus\{8\}\subseteq A\cup\{4,5\}$ has size $4$, and having at most two vertices in $A$ forces $\{4,5\}\subseteq S$.
Thus $\{4,5,8\}\subseteq S$ and $S$ contains exactly two vertices from $A$.
However, the edges of $H_0$ containing $8$ are
\[
\{3,4,8\},\ \{3,5,8\},\ \{4,6,8\},\ \{5,7,8\}.
\]
Since $4,5\in S$, avoiding an edge forces $3,6,7\notin S$, so $S\cap A\subseteq\{2\}$, contradicting that $S$ contains two vertices from $A$.
Hence $S$ must contain an edge of $H_0$.
The case $9\in S$ and $8\notin S$ is analogous (using the edges $\{2,4,9\},\{2,5,9\},\{4,6,9\},\{5,7,9\}$).

\smallskip
\noindent\emph{Case 3: $S\cap\{8,9\}=\{8,9\}$.}
Let $T:=S\setminus\{8,9\}$, so $|T|=3$ and $T\subseteq A\cup\{4,5\}$.
If $T\subseteq A$, then $T$ is a $3$-subset of $A$ and hence an edge of $H_0$.
So assume $T$ contains $4$ or $5$.

If $4\in T$ and $S$ contains no edge, then $T$ cannot contain $2$ (else $\{2,4,9\}\subseteq S$), cannot contain $3$ (else $\{3,4,8\}\subseteq S$), and cannot contain $6$ (else $\{4,6,8\}\subseteq S$ or $\{4,6,9\}\subseteq S$).
Thus $T=\{4,5,7\}$, which yields $\{5,7,8\}\subseteq S$, an edge of $H_0$.
Similarly, if $5\in T$, then to avoid edges we must have $T=\{4,5,6\}$, which yields $\{4,6,8\}\subseteq S$, again an edge.
Therefore $S$ contains an edge of $H_0$ in this case as well.

This proves that every $5$-subset of $\{2,\dots,9\}$ contains an edge of $H_0$.

\smallskip
\noindent\emph{Step 3: conclusion.}
By \eqref{eq:B-size}, the set $B$ has size at least $5$, hence $B$ contains an edge of $H_0$ by Step~2.
But all vertices in $B$ have colour $1$ by definition, so that edge is monochromatic~$1$, contradicting that $\varphi$ is a proper $2$-colouring.
\end{proof}

\subsection{\texorpdfstring{$H$}{H} is \texorpdfstring{$3$}{3}-colourable}

\begin{lemma}\label{lem:3-colorable}
The hypergraph $H$ admits a proper $3$-colouring.
Hence $\chi(H)=3$.
\end{lemma}

\begin{proof}
Define $\psi:\V(H)\to\{1,2,3\}$ by
\[
\begin{aligned}
\psi(1)&=\psi(2)=\psi(4)=\psi(5)=1,\\
\psi(3)&=\psi(6)=\psi(8)=\psi(9)=2,\\
\psi(7)&=3.
\end{aligned}
\]
One checks directly from \eqref{eq:H-edges} that no edge is monochromatic under $\psi$.
By Lemma~\ref{lem:not-2-colorable}, $H$ is not $2$-colourable, so $\chi(H)=3$.
\end{proof}

\subsection{Edge-criticality}

\begin{proposition}\label{prop:edge-critical}
For every edge $e\in \E(H)$ we have $\chi(H-e)\le 2$.
\end{proposition}

\begin{proof}
By Lemma~\ref{lem:edge-cert}, it suffices to exhibit, for each $e\in \E(H)$, a $2$-colouring of $\V(H)$ in which $e$ is the unique monochromatic edge.
Such certificates are listed explicitly in Appendix~\ref{app:certificates} (Table~\ref{tab:edge-cert}).
\end{proof}

\subsection{Vertex-criticality}

\begin{proposition}\label{prop:vertex-critical}
For every vertex $v\in \V(H)$ we have $\chi(H-v)\le 2$.
\end{proposition}

\begin{proof}
For each $v\in \V(H)$, Appendix~\ref{app:certificates} (Table~\ref{tab:vertex-cert}) lists an explicit proper $2$-colouring of $H-v$.
\end{proof}

\begin{proof}[Proof of Theorem~\ref{thm:intro-chromatic}]
Let $H$ be the hypergraph defined in \eqref{eq:H-edges}.
By Lemmas~\ref{lem:not-2-colorable} and~\ref{lem:3-colorable}, we have $\chi(H)=3$.
By Propositions~\ref{prop:edge-critical} and~\ref{prop:vertex-critical}, deleting any edge or any vertex makes $H$ $2$-colourable.
Thus $H$ is critically $3$-chromatic in the sense of Definition~\ref{def:critical-3chrom}.
Finally, Lemma~\ref{lem:degrees} gives $\del(H)=7$.
\end{proof}

\section{Concluding remarks and reproducibility}\label{sec:concluding}

The Erd\H{o}s--Lov\'asz question on minimum degree $7$ for ``$3$-critical'' $3$-uniform hypergraphs splits naturally into two non-equivalent interpretations that both occur in the literature and are explicitly highlighted in modern discussions.\cite{ErdosProblems834}
In the transversal sense, Section~\ref{sec:tau} shows that the problem has a sharp negative answer: a $\tau$-critical $3$-graph of order $3$ has at most $10$ edges and therefore satisfies $\del\le 6$, with equality attained by $K^{(3)}_5$.
In the chromatic sense, Section~\ref{sec:chromatic} gives a sharp positive answer by exhibiting an explicit critically $3$-chromatic $3$-graph on $9$ vertices with minimum degree $7$.

Finally, we record the reproducibility aspects of the chromatic-critical construction.
Appendix~\ref{app:edges} provides the edge set of $H$ in both mathematical and machine-readable form.
Appendix~\ref{app:certificates} lists explicit $2$-colouring certificates for each edge deletion and each vertex deletion, which constitute a complete finite proof of Propositions~\ref{prop:edge-critical} and~\ref{prop:vertex-critical}.
Appendix~\ref{app:verification} contains a short Python~3 script (using only the standard library) that exhaustively checks all claims about $H$ by brute force over the $2^9$ possible $2$-colourings, and also validates every certificate listed in Appendix~\ref{app:certificates}.
This computation is not logically required, but it provides an additional layer of protection against transcription errors and makes the example straightforward to reuse in further work.

\clearpage
\appendix

\section{The edge set of \texorpdfstring{$H$}{H}}\label{app:edges}

For convenience, we reproduce the edge set $\E(H)$ from \eqref{eq:H-edges} in a machine-readable format suitable for direct input into a verification script.

\subsection*{Machine-readable edge list}

\begin{lstlisting}
1 2 3
1 2 9
1 3 8
1 4 6
1 4 8
1 4 9
1 5 7
1 5 8
1 5 9
1 6 7
2 3 6
2 3 7
2 4 9
2 5 9
2 6 7
3 4 8
3 5 8
3 6 7
4 6 8
4 6 9
5 7 8
5 7 9
\end{lstlisting}

\section{Certificates for edge- and vertex-criticality}\label{app:certificates}

We represent a $2$-colouring $\varphi:\V(H)\to\{0,1\}$ by its \emph{blue set}
\[
B(\varphi):=\{v\in \V(H): \varphi(v)=1\}.
\]
Vertices not in $B(\varphi)$ are coloured $0$ (red).

\subsection{Edge-deletion certificates}

For each edge $e\in \E(H)$, Table~\ref{tab:edge-cert} lists a blue set $B(\varphi)$ such that $e$ is the unique monochromatic edge of $H$ under the corresponding colouring $\varphi$.
By Lemma~\ref{lem:edge-cert}, this certifies $\chi(H-e)\le 2$ for every $e$.

\begin{longtable}{@{}ll@{}}
\caption{Edge-deletion certificates: $B(\varphi)$ makes $e$ the unique monochromatic edge.}\label{tab:edge-cert}\\
\toprule
Edge $e$ & Blue vertices $B(\varphi)$\\
\midrule
\endfirsthead
\toprule
Edge $e$ & Blue vertices $B(\varphi)$\\
\midrule
\endhead
\midrule
\multicolumn{2}{r}{\small Continued on next page}\\
\endfoot
\bottomrule
\endlastfoot
$\{1,2,3\}$ & $\{6,7,8,9\}$ \\
$\{1,2,9\}$ & $\{3,4,5,6\}$ \\
$\{1,3,8\}$ & $\{2,4,5,6\}$ \\
$\{1,4,6\}$ & $\{2,7,8,9\}$ \\
$\{1,4,8\}$ & $\{3,5,6,9\}$ \\
$\{1,4,9\}$ & $\{2,5,6,8\}$ \\
$\{1,5,7\}$ & $\{2,6,8,9\}$ \\
$\{1,5,8\}$ & $\{3,4,7,9\}$ \\
$\{1,5,9\}$ & $\{2,4,7,8\}$ \\
$\{1,6,7\}$ & $\{2,3,4,5\}$ \\
$\{2,3,6\}$ & $\{1,7,8,9\}$ \\
$\{2,3,7\}$ & $\{1,6,8,9\}$ \\
$\{2,4,9\}$ & $\{1,3,5,6\}$ \\
$\{2,5,9\}$ & $\{1,3,4,7\}$ \\
$\{2,6,7\}$ & $\{1,3,4,5\}$ \\
$\{3,4,8\}$ & $\{1,2,5,6\}$ \\
$\{3,5,8\}$ & $\{1,2,4,7\}$ \\
$\{3,6,7\}$ & $\{1,2,4,5\}$ \\
$\{4,6,8\}$ & $\{1,3,7,9\}$ \\
$\{4,6,9\}$ & $\{1,2,7,8\}$ \\
$\{5,7,8\}$ & $\{1,3,6,9\}$ \\
$\{5,7,9\}$ & $\{1,2,6,8\}$ \\
\end{longtable}

\subsection{Vertex-deletion certificates}

For each vertex $v\in \V(H)$, Table~\ref{tab:vertex-cert} lists a blue set $B(\varphi)\subseteq \V(H)\setminus\{v\}$ defining a proper $2$-colouring of $H-v$.

\begin{table}[t]
\centering
\small
\begin{tabular}{@{}ll@{}}
\toprule
Deleted vertex $v$ & Blue vertices $B(\varphi)$ in $H-v$\\
\midrule
1 & $\{2,3,4,5\}$ \\
2 & $\{1,3,4,5\}$ \\
3 & $\{1,2,4,5\}$ \\
4 & $\{1,2,5,6\}$ \\
5 & $\{1,2,4,7\}$ \\
6 & $\{1,2,4,5\}$ \\
7 & $\{1,2,4,5\}$ \\
8 & $\{1,2,4,7\}$ \\
9 & $\{1,2,6,8\}$ \\
\bottomrule
\end{tabular}
\caption{Vertex-deletion certificates: $B(\varphi)$ gives a proper $2$-colouring of $H-v$.}
\label{tab:vertex-cert}
\end{table}

\section{Verification script}\label{app:verification}

The following Python~3 script (standard library only) verifies:
(i) the minimum degree $\del(H)=7$;
(ii) $H$ is not $2$-colourable;
(iii) $H-e$ is $2$-colourable for every edge $e$;
(iv) $H-v$ is $2$-colourable for every vertex $v$;
and (v) all certificates in Appendix~\ref{app:certificates}.

\begin{lstlisting}
# verify_critical_hypergraph.py
# Pure Python 3, no dependencies.

from itertools import product

V = list(range(1, 10))

E = [
    (1, 2, 3), (1, 2, 9), (1, 3, 8), (1, 4, 6), (1, 4, 8), (1, 4, 9),
    (1, 5, 7), (1, 5, 8), (1, 5, 9), (1, 6, 7),
    (2, 3, 6), (2, 3, 7), (2, 4, 9), (2, 5, 9), (2, 6, 7),
    (3, 4, 8), (3, 5, 8), (3, 6, 7),
    (4, 6, 8), (4, 6, 9),
    (5, 7, 8), (5, 7, 9),
]

# Edge-deletion certificates from Table B.1 (blue sets)
EDGE_CERTS = {
    (1, 2, 3): {6, 7, 8, 9},
    (1, 2, 9): {3, 4, 5, 6},
    (1, 3, 8): {2, 4, 5, 6},
    (1, 4, 6): {2, 7, 8, 9},
    (1, 4, 8): {3, 5, 6, 9},
    (1, 4, 9): {2, 5, 6, 8},
    (1, 5, 7): {2, 6, 8, 9},
    (1, 5, 8): {3, 4, 7, 9},
    (1, 5, 9): {2, 4, 7, 8},
    (1, 6, 7): {2, 3, 4, 5},
    (2, 3, 6): {1, 7, 8, 9},
    (2, 3, 7): {1, 6, 8, 9},
    (2, 4, 9): {1, 3, 5, 6},
    (2, 5, 9): {1, 3, 4, 7},
    (2, 6, 7): {1, 3, 4, 5},
    (3, 4, 8): {1, 2, 5, 6},
    (3, 5, 8): {1, 2, 4, 7},
    (3, 6, 7): {1, 2, 4, 5},
    (4, 6, 8): {1, 3, 7, 9},
    (4, 6, 9): {1, 2, 7, 8},
    (5, 7, 8): {1, 3, 6, 9},
    (5, 7, 9): {1, 2, 6, 8},
}

# Vertex-deletion certificates from Table B.2 (blue sets on V\{v})
VERTEX_CERTS = {
    1: {2, 3, 4, 5},
    2: {1, 3, 4, 5},
    3: {1, 2, 4, 5},
    4: {1, 2, 5, 6},
    5: {1, 2, 4, 7},
    6: {1, 2, 4, 5},
    7: {1, 2, 4, 5},
    8: {1, 2, 4, 7},
    9: {1, 2, 6, 8},
}

def degs(edges):
    d = {v: 0 for v in V}
    for a, b, c in edges:
        d[a] += 1
        d[b] += 1
        d[c] += 1
    return d

def is_mono(edge, col):
    a, b, c = edge
    return col[a] == col[b] == col[c]

def induced_edges(edges, verts):
    verts = set(verts)
    return [e for e in edges if set(e) <= verts]

def proper_2coloring_exists(edges, verts):
    verts = list(verts)
    for bits in product([0, 1], repeat=len(verts)):
        col = {v: bits[i] for i, v in enumerate(verts)}
        if all(not is_mono(e, col) for e in edges):
            return True, col
    return False, None

def monochromatic_edges(edges, col):
    return [e for e in edges if is_mono(e, col)]

def main():
    # Degree check
    d = degs(E)
    print("degrees:", d)
    print("min degree:", min(d.values()))
    assert min(d.values()) == 7

    # Check H not 2-colorable
    ok, _ = proper_2coloring_exists(E, V)
    print("H 2-colorable?", ok)
    assert not ok

    # Check edge-criticality by brute force
    for e in E:
        edges2 = [f for f in E if f != e]
        ok, _ = proper_2coloring_exists(edges2, V)
        if not ok:
            raise AssertionError(f"H-e not 2-colorable for e={e}")
    print("Edge-criticality verified by brute force.")

    # Check vertex-criticality by brute force
    for v in V:
        verts2 = [u for u in V if u != v]
        edges2 = induced_edges(E, verts2)
        ok, _ = proper_2coloring_exists(edges2, verts2)
        if not ok:
            raise AssertionError(f"H-v not 2-colorable for v={v}")
    print("Vertex-criticality verified by brute force.")

    # Validate edge certificates
    for e, B in EDGE_CERTS.items():
        col = {v: (1 if v in B else 0) for v in V}
        monos = monochromatic_edges(E, col)
        if monos != [e]:
            raise AssertionError(f"Bad edge certificate for e={e}: mono edges={monos}")
    print("All edge certificates validated.")

    # Validate vertex certificates
    for v, B in VERTEX_CERTS.items():
        verts2 = [u for u in V if u != v]
        edges2 = induced_edges(E, verts2)
        col = {u: (1 if u in B else 0) for u in verts2}
        if any(is_mono(e, col) for e in edges2):
            raise AssertionError(f"Bad vertex certificate for v={v}")
    print("All vertex certificates validated.")

    print("All checks passed.")

if __name__ == "__main__":
    main()
\end{lstlisting}

\end{document}